\documentclass[conference]{IEEEtran}

\IEEEoverridecommandlockouts

\usepackage{amsmath,amsthm,relsize}
\usepackage{amsfonts, amssymb, cuted}
\usepackage{mathtools}
\usepackage{algorithm}
\usepackage[noend]{algpseudocode}
\usepackage{cite}
\usepackage{xcolor}
\usepackage{soul}
\usepackage{graphicx}
\usepackage{pgfplotstable}
\usepackage{pgfplots}
\usepackage[font=footnotesize]{caption}
\usepackage{multirow}
\pgfplotsset{compat=1.14}

\newtheorem{lem}{Lemma}

\theoremstyle{definition}
\newtheorem{exmp}{Example}

\newcommand{\subparagraph}{}
\usepackage{titlesec}

\begin{document}

\title{Subpacketization-Rate Trade-off\\in Multi-Antenna Coded Caching}
\author{
\IEEEauthorblockN{
MohammadJavad Salehi\IEEEauthorrefmark{1},
Antti T\"olli\IEEEauthorrefmark{1},
Seyed Pooya Shariatpanahi\IEEEauthorrefmark{2},
Jarkko Kaleva\IEEEauthorrefmark{1}
}
\\
\IEEEauthorblockA{
\IEEEauthorrefmark{1}Center for Wireless Communications, University of Oulu, Oulu, Finland.}
\IEEEauthorblockA{
\IEEEauthorrefmark{2}School of Computer Science, Institute for Research in Fundamental Sciences (IPM), Tehran, Iran.}
\IEEEauthorblockA{\{fist\_name.last\_name\}@oulu.fi, pooya@ipm.ir
\vspace{-4ex}}
	\thanks{This work was supported by the Academy of Finland under grants no. 319059 (Coded Collaborative Caching for Wireless Energy Efficiency) and 318927 (6Genesis Flagship).}
}

\maketitle
\begin{abstract}
Coded caching can be applied in wireless multi-antenna communications by multicast beamforming coded data chunks to carefully selected user groups and using the existing file fragments in user caches to decode the desired files at each user. However, the number of packets a file should be split into, known as subpacketization, grows exponentially with the network size. We provide a new scheme, which enables the level of subpacketization to be selected freely among a set of predefined values depending on basic network parameters such as antenna and user count. A simple efficiency index is also proposed as a performance indicator at various subpacketization levels. The numerical examples demonstrate that larger subpacketization generally results in better efficiency index and higher symmetric rate, while smaller subpacketization incurs significant loss in the achievable rate. This enables more efficient caching schemes, tailored to the available computational and power resources.
\end{abstract}
\begin{IEEEkeywords}
multi-antenna coded caching,
multicast beamforming,
flexible subpacketization
\end{IEEEkeywords}

\section{Introduction}
Global mobile data traffic has been subject to consistent strong growth during the recent years. Cisco visual networking index \cite{index2019global} estimates the global mobile traffic to reach 77 Exabytes per month by 2022, more than six-fold increase compared to 2017. On the other hand, more than 80\% of this number is expected to be generated by mobile video, and as a result research efforts in efficient video delivery have been intensified in recent years. Caching content closer to the end-users, known as \textit{Edge Caching}, is among the promising solutions proposed in this direction \cite{bastug2014living, wang2014cache}. Coded caching, first appeared in the pioneering work of \cite{maddah2014fundamental}, is an interesting approach to edge caching for wireless communications, through which one can take benefit of the broadcast nature of wireless channels. It is shown in \cite{maddah2014fundamental} that caching at end-user locations and multicasting coded file chunks to carefully selected user groups enables reduction in delivery bandwidth over the broadcast channel.

The original scheme of \cite{maddah2014fundamental} assumes a separate, offline cache placement phase and an error-free broadcast channel to the end-users. Since then, significant effort has been carried out by the research community to apply the coded caching concept to more general network setups. For example, in \cite{maddah2015decentralized,pedarsani2016online,karamchandani2016hierarchical} decentralized, online and hierarchical coded caching are considered, respectively. Similarly, extension of \cite{maddah2014fundamental} to multi-server scenario is studied in \cite{shariatpanahi2016multi}, and the same concept is then applied to multi-antenna wireless communications in \cite{shariatpanahi2017multi, tolli2017multi, tolli2018linear, lampiris2018adding}.

We focus on the application of coded caching in multi-antenna wireless communications and the well-known problem of subpacketization, which means the number of packets each file should be split into, for a caching scheme to work properly \cite{lampiris2018adding}. The very large subpacketization required by schemes presented in \cite{shariatpanahi2017multi, tolli2017multi, tolli2018linear} makes them infeasible for even moderate-sized networks, and at the same time adds considerably to the complexity of beamformer design. Reducing subpacketization is partially tackled in \cite{lampiris2018adding}, by assigning users to fixed size groups and using zero-forcing (ZF) and finite field summation for intra- and inter-group interference cancellation, respectively. However, as will be discussed, this scheme also results in rate loss compared to \cite{tolli2017multi,tolli2018linear}.

In this paper, we provide a new coded caching scheme, which enables the selection of a desired subpacketization level among a set of possible values depending on basic network parameters such as user count, number of available antennas and the global cache ratio (as defined in \cite{maddah2014fundamental}). We show that the schemes provided in \cite{shariatpanahi2017multi} and \cite{lampiris2018adding} are two extreme cases of our scheme, built with the highest and lowest possible subpacketization levels, respectively. We also provide a simple efficiency index, which is an indicator of how well the performance (in terms of symmetric rate) will be for a specific subpacketization level. We 
show that larger subpacketization usually results in higher efficiency index; and higher communication rate accordingly. Particularly, the schemes of \cite{tolli2017multi} and \cite{lampiris2018adding} provide the highest and lowest symmetric rate, and the rate difference is more considerable at low-SNR regime.


In this paper, we use the following notations. Boldface lower case and capital letters indicate vectors and matrices, respectively. Sets are shown with calligraphic letters. For a positive integer $K$, $[K]$ means the set $\{1,2,...,K \}$. For two sets $\mathcal{S}$ and $\mathcal{T}$, $|\mathcal{S}|$ indicates the number of elements in $\mathcal{S}$ and $\mathcal{S} \backslash \mathcal{T}$ is the set of elements in $\mathcal{S}$ which are not in $\mathcal{T}$. For any vector $\mathbf{u}$, $\| \mathbf{u} \|$ is the second norm of $\mathbf{u}$. 

\section{System Model and Literature Review}
\label{Section:system_model}
We consider $K$ users in a MISO wireless channel with $L$ antennas at the transmitter. Each user has a cache memory of size $MF$ bits. The cache contents are updated during the placement phase, which can be for example during off-peak hours. During the delivery phase, users request files from a library of size $N$, where each file is of size $F$ bits. In order to fulfill the requests, each requesting node uses its cache content together with the data received over the communication channel. The problem is to design placement and delivery phases, such that the time needed for all users to decode their requested files is minimized; or equivalently, the symmetric rate of all users is maximized.

In \cite{maddah2014fundamental} it is shown that a global caching gain, proportional to the total cache size available in the network, is achievable in addition to the local gain of each single cache memory. This added gain is obtained by multicasting coded data chunks to multiple cache-enabled users, so we use the term \textit{multicasting gain} interchangeably with global caching gain in this paper.

Defining global cache ratio $t = \frac{KM}{N}$, the scheme in \cite{maddah2014fundamental} provides a multicasting gain of $t+1$. In \cite{shariatpanahi2016multi} it is shown that using $L$ servers simultaneously, the multicasting gain increases to $t+L$. In \cite{shariatpanahi2017multi} this result is applied to the wireless communication scenario with $L$ antennas at the transmitter. Here we briefly review the scheme of \cite{shariatpanahi2017multi}, as it shares many common parts with the new scheme presented in this paper.

The placement phase of \cite{shariatpanahi2017multi} requires each file $W$ to be first split into $\binom{K}{t}$ equal-sized packets $W_{p(\mathcal{T})}$, where $p(\mathcal{T})$ assigns a unique index for each $\mathcal{T} \subseteq [K]$ with $|\mathcal{T}|=t$. Each packet is then further split into 
\begin{equation}
\label{eq:Q_def}
    Q = \binom{K-t-1}{L-1}
\end{equation}
equal-sized subpackets $W_{p(\mathcal{T})}^q$; and each $W_{p(\mathcal{T})}^q$ is stored at the cache memory of all users $k \in \mathcal{T}$. 
During the delivery phase, for each $\mathcal{S} \subseteq [K]$ with $|\mathcal{S}|=t+L$ a separate vector $\mathbf{x}(\mathcal{S})$ is transmitted. Assume user $k$ requests the file $W(k)$. For simplicity, we ignore the time index for vectors $\mathbf{x}(\mathcal{S})$ sent in TDMA fashion. To create $\mathbf{x}(\mathcal{S})$, first for each $\mathcal{V} \subseteq \mathcal{S}$ with $|\mathcal{V}| = t+1$, codeword $X(\mathcal{V})$ is built as
\begin{equation}
    X(\mathcal{V}) = \bigoplus \limits_{k \in \mathcal{V}} W_{p(\mathcal{V} \backslash \{k\})} ^ {q(k,\mathcal{V})} (k) \; ,
\end{equation}
where $\oplus$ denotes bit-wise XOR and $q(k,\mathcal{V})$ is defined such that it is guaranteed that each subpacket is transmitted only once and also after all transmissions are concluded, each user is able to decode its requested file (c.f. \cite{shariatpanahi2017multi}). Next, vector $\mathbf{u}(\mathcal{V})$ is defined such that $\| \mathbf{u}(\mathcal{V}) \|  = 1$ and
\begin{equation}
    \mathbf{u}(\mathcal{V}) \perp \mathbf{h}_k \qquad \forall k \in \mathcal{S} \backslash \mathcal{V} \; ,
\end{equation}
where $\perp$ denotes perpendicularity and $\mathbf{h}_k$ is the $L \times 1$ channel vector from base station to user $k$. Finally, $\mathbf{x}(\mathcal{S})$ is built as
\begin{equation}
\label{eq:transmission_vector_power_optimization}
    \mathbf{x}(\mathcal{S}) = \sum_{\mathcal{V}\subseteq \mathcal{S}}  \mathbf{u}(\mathcal{V})p(\mathcal{V})X(\mathcal{V}) \; ,
\end{equation}
where $p(\mathcal{V})$ is the power allocated to the transmission of codeword $X(\mathcal{V})$. It can be verified that by transmitting $\mathbf{x}(\mathcal{S})$, each user $k \in \mathcal{S}$ is able to decode part of $W(k)$ using its cache contents, and after concluding successive transmissions for all $\mathcal{S} \subseteq [K]$, all users in the system are able to completely decode their requested files. Specifically, after a single transmission $\mathbf{x}(\mathcal{S})$ is concluded, each user $k \in \mathcal{S}$ faces a MAC channel with $\binom{t+L-1}{t}$ terms, and the transmission rate and power should be adjusted such that simultaneous decoding of all terms is possible at every user.

Although the scheme of \cite{shariatpanahi2017multi} achieves the full multicasting gain of $t+L$ for each transmission, it has major drawbacks. First, the required subpacketization grows exponentially with $K$, making the scheme infeasible for even moderate values of $K$ \cite{lampiris2018adding}. Second, ZF results in poor rate specially in low-SNR regime \cite{tolli2018linear}. Finally, the complexity of successive interference cancellation for decoding all terms in the MAC channel grows exponentially with the MAC size \cite{tolli2017multi}. Although workarounds for drawbacks are available in the literature \cite{lampiris2018adding,tolli2018linear,tolli2017multi}, unfortunately improvements for any drawback causes the others to become worse. As a quick review, the scheme provided in \cite{lampiris2018adding} divides users into groups of size $L$ and uses ZF and finite field summation to mitigate both intra- and inter-group interference. Although this solution results in a dramatic drop in the required subpacketization (specially if $\frac{t}{L}$ is an integer), it provides a lower symmetric rate than other schemes, as demonstrated later in Section \ref{Section:Simulation_Results}. On the other hand, the scheme of \cite{tolli2018linear} suggests designing optimized multicast beamformers instead of ZF to improve the performance at low- and mid-SNR. More precisely, instead of \eqref{eq:transmission_vector_power_optimization} $\mathbf{x}(\mathcal{S})$ is created as
\begin{equation}
\label{eq:transmission_vector_beamformer}
    \mathbf{x}(\mathcal{S}) = \sum_{\mathcal{V}\subseteq \mathcal{S}}  \mathbf{w}(\mathcal{V})X(\mathcal{V}) \; ,
\end{equation}
where $\mathbf{w}(\mathcal{V})$ is the general beamformer vector. The unwanted terms then appear as (Gaussian) interference instead of being nulled at each user. Although this formulation significantly improves the rate at low-SNR, considering interference between parallel multi-group multicast messages makes the beamformer design significantly more complex \cite{tolli2017multi}. Finally, different schemes with reduced complexity are proposed in \cite{tolli2017multi}. However, in certain scenarios, this may have further detrimental impact on the required subpacketization. 

In this paper we provide a new scheme, achieving full multicasting gain of $t+L$ with subpacketization $P \times Q$, where $Q$ is defined in \eqref{eq:Q_def} and $P$ can be selected freely among a set of predefined values.
We show that the schemes of \cite{shariatpanahi2017multi} and \cite{lampiris2018adding} are two extreme cases of this new scheme, with $P=\binom{K}{t}$ and $P=\binom{K/L}{t/L}$ respectively\footnote{For non-integer $K/t$ and $t/L$, $\lfloor K/t \rfloor$ and $\lceil t/L \rceil$ are used.}.

\begin{figure}
    \centering
    \resizebox{0.9\columnwidth}{!}{%
    
    \begin{tikzpicture}

    \begin{axis}
    [
    axis lines = left,
    xlabel = \smaller {SNR [dB]},
    ylabel = \smaller {Symmetric Rate [nats/s]},
    ylabel near ticks,
    legend pos = south east,
    ticklabel style={font=\smaller},
    grid=both,
    major grid style={line width=.2pt,draw=gray!30},
    ]
    
    \addplot
    [
    mark = +,
    black
    ]
    table
    [
    y=P6,
    x=SNR
    ]
    {Data/Exmp_K4_Plot.txt};
    \addlegendentry{\smaller 6 Subpackets}
    
    \addplot
    [
    mark = x,
    gray
    ]
    table
    [
    y=P4,
    x=SNR
    ]
    {Data/Exmp_K4_Plot.txt};
    \addlegendentry{\smaller 4 Subpackets}
    
    \addplot
    [
    dashed,
    mark = +,
    black
    ]
    table
    [
    y=P2,
    x=SNR
    ]
    {Data/Exmp_K4_Plot.txt};
    \addlegendentry{\smaller 2 Subpackets}

    \end{axis}

    \end{tikzpicture}
    }
    \caption{Rate vs SNR, Optimized Beamformer - $K=4$, $t=L=2$}
    \label{fig:exmp_k4_plot}
\end{figure}
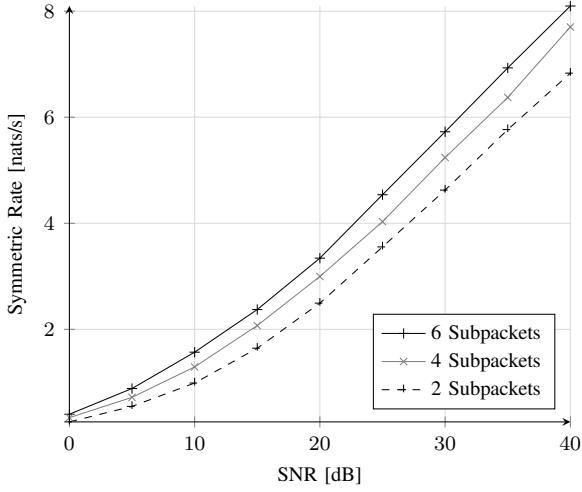

Before providing the details, let us consider a network with $K=4$ and $L=t=2$. The scheme provided in \cite{tolli2017multi,tolli2018linear} requires subpacketization $P = \binom{4}{2} = 6$, while the one presented in \cite{lampiris2018adding} reduces subpacketization to $P = \binom{2}{1} = 2$. Using the new scheme, for this network one can build another coded caching scheme with subpacketization $P = 4$. Comparison of the rate versus SNR for these three schemes is shown in Figure~\ref{fig:exmp_k4_plot}. It should be noted that the results in Figure~\ref{fig:exmp_k4_plot} are calculated using optimized beamformer design, which provides better rate than ZF used in \cite{lampiris2018adding}. Still, it is clear that increasing $P$ results in an increase in the symmetric rate. 
The rate advantage of $P=4$ over $P=2$ is $32\%$ at 0 dB, eventually decreasing to $13\%$ at 40 dB. This illustrates the stronger impact of increased subpacketization at lower SNR.

\section{Subpacketization-Selective Scheme}
\label{Section:trade_off}
First we consider the special case of $K=t+L$, and then provide an extension to more general setups.

\subsection{Cache Placement Scheme}
\label{subsec:placement_matrix}
Assume $K=t+L$ and each file $W$ is divided into $P$ equal-sized parts $W_p$. We define a binary placement matrix $\mathbf{V}$ of size $P \times K$ and represent its elements by $v_{p,k}$, where $p \in [P]$ and $k \in [K]$. If $v_{p,k} = 1$, then for any file $W$ in the file library, we store $W_p$ in the cache memory of user $k$. For example consider the three curves in Figure \ref{fig:exmp_k4_plot}. The corresponding placement matrix for $P=2$ is
\begin{equation}
\label{eq:V_mat_4_2_Petros}
    \mathbf{V} = 
    \begin{bmatrix}
        1 & 0 & 1 & 0 \\
        0 & 1 & 0 & 1
    \end{bmatrix} \; ,
\end{equation}
while for $P=4$ and $P=6$ we have used
\begin{equation}
\label{eq:V_mat_4_2_Full}
    \mathbf{V}' = 
    \begin{bmatrix}
    1 & 1 & 0 & 0 \\
    0 & 1 & 1 & 0 \\
    0 & 0 & 1 & 1 \\
    1 & 0 & 0 & 1
    \end{bmatrix}, \quad
    \mathbf{V}'' = 
    \begin{bmatrix}
        1 & 1 & 0 & 0 \\
        0 & 1 & 1 & 0 \\
        0 & 0 & 1 & 1 \\
        1 & 0 & 0 & 1 \\
        1 & 0 & 1 & 0 \\
        0 & 1 & 0 & 1
    \end{bmatrix} \; ,
\end{equation}
respectively. Denoting the cache content of user $k$ with $Z(k)$ and considering the file library $\{A,B,C,D\}$, we have
\begin{equation*}
    \begin{aligned}
        Z(1) &= \{A_1, B_1, C_1, D_1 \} \\ 
        Z'(1) &=  \{A_1,A_4,B_1,B_4,C_1,C_4,D_1,D_4 \} \\
        Z''(1) &= \{A_1, A_4, A_5, B_1, B_4, B_5, C_1, C_4, C_5, D_1, D_4, D_5\}
    \end{aligned}
\end{equation*}
corresponding to $\mathbf{V}$, $\mathbf{V}'$ and $\mathbf{V''}$, respectively. We say $\mathbf{V}$ is a \textit{valid} placement matrix, if its rows are different and 
\begin{subequations}
\begin{align}
    \label{eq:per_user_cache_size} \sum_{p} v_{p,k} = \frac{Pt}{K} \qquad & \forall k \in [K] \\
    \label{eq:per_part_cache_size} \sum_{k} v_{p,k} = t \qquad & \forall p \in [P]
\end{align}
\end{subequations}

Note that \eqref{eq:per_user_cache_size} represents the cache size constraint at each user; i.e. as the size of each packet is $\frac{F}{P}$ and there exists $N$ files in the library, the required cache size at each user becomes

\begin{equation}
    \sum_{p} v_{p,k} \times N \times \frac{F}{P} = \frac{P}{K} \times \frac{KM}{N} \times \frac{NF}{P} = MF \; ,
\end{equation}
which is exactly the available cache size at each user. Moreover, \eqref{eq:per_part_cache_size} guarantees that each packet is replicated for the same number of times throughout all cache memories in the network. For example, it can be easily checked that all the placement matrices in \eqref{eq:V_mat_4_2_Petros} and \eqref{eq:V_mat_4_2_Full} are valid.
According to \eqref{eq:per_user_cache_size} and \eqref{eq:per_part_cache_size} it is clear that reordering rows or columns in a valid placement matrix results in another valid placement matrix. The following lemmas represent the relationship between valid placement matrices and caching schemes of \cite{shariatpanahi2017multi} and \cite{lampiris2018adding}.

\begin{lem}
\label{lem:full_placement}
If $K=t+L$, the cache placement scheme of \cite{shariatpanahi2017multi} is equivalent to a valid placement matrix with $P = \binom{K}{t}$.
\end{lem}

\begin{proof}
Consider a valid placement matrix $\mathbf{V}$ with $P = \binom{K}{t}$ rows and $K$ columns. According to the validity condition \eqref{eq:per_part_cache_size}, there exist exactly $t$ non-zero elements at each row of $\mathbf{V}$. As a result, the rows of $\mathbf{V}$ include all possible combinations of $t$ non-zero elements at $K$ positions. Labelling each row with a set $\mathcal{T}$ consisting of column indices of non-zero elements in that row and using the same index $p(\mathcal{T})$ for packets as described in section \ref{Section:system_model}, we reach the same placement of \cite{shariatpanahi2017multi}.
\end{proof}

\begin{lem}
\label{lem:petros_placement}
If $\frac{K}{L}$ and $\frac{t}{L}$ are integers, then the cache placement scheme of \cite{lampiris2018adding} is equivalent to a valid placement matrix with $P=\binom{K/L}{t/L}$.
\end{lem}

\begin{proof}
Similar to Lemma \ref{lem:full_placement}, follows from labeling rows in a valid placement matrix with $P=\binom{K/L}{t/L}$.
\end{proof}

Lemmas \ref{lem:full_placement} and \ref{lem:petros_placement} provide higher and lower boundaries for subpacketization, respectively. Any other $\binom{K/L}{t/L} < P < \binom{K}{t}$, for which validity conditions hold, provides a possible cache placement with a different subpacketization value.

\subsection{Delivery Scheme}
\label{subsec:delivery_scheme}
Assume cache placement is performed according to a given valid placement matrix $\mathbf{V}$. At the start of the delivery phase, each user $k$ reveals its requested file $W(k)$. 
For each $\mathcal{V} \subseteq [K]$ with $|\mathcal{V}| = t+1$ we create a codeword $X(\mathcal{V})$ such that all users in $\mathcal{V}$ can decode part of their missing data (the parts for which the corresponding $v_{p,k}$ is zero) with the help of $X(\mathcal{V})$. Defining the \textit{corresponding packet set} of $\mathcal{V}$ as
\begin{equation}
\label{eq:corresponding_packet_set}
    \Phi(\mathcal{V}) = \{ p \in [P] \quad \mid \quad v_{p,k} = 0, \; \forall k \in [K] \backslash \mathcal{V} \} \; ,
\end{equation}
the codeword $X(\mathcal{V})$ is built as
\begin{equation}
\label{eq:codeword_base}
    X(\mathcal{V}) = \bigoplus \limits_{\substack{k \in \mathcal{V} \\ p \in \Phi(\mathcal{V})}} (1-v_{p,k})W_p(k) \; .
\end{equation}

As a brief explanation, according to the placement condition~\eqref{eq:per_part_cache_size}, each row of $\mathbf{V}$ has exactly $t$ non-zero and $K-t=L$ zero elements. Also as we have $L$ antennas, it is possible to null (or suppress) $X(\mathcal{V})$ at $L-1$ users. The set $\Phi(\mathcal{V})$ contains packet indices for which the corresponding $v_{p,k}$ is zero for all $L-1$ users in $[K] \backslash \mathcal{V}$, and so for each $p \in \Phi(\mathcal{V})$, there exists exactly one $k \in \mathcal{V}$ such that $v_{p,k} = 0$. This enables each user $k \in \mathcal{V}$ to remove unwanted terms from $X(\mathcal{V})$ using its cache contents, and decode its missing packet accordingly.

Nulling (or suppressing) each $X(\mathcal{V})$ at all $L-1$ users in $[K] \backslash \mathcal{V}$ also enables us to transmit all codewords simultaneously, thus achieving the full multicasting gain of $t+L$.
Defining $\mathbf{u}(\mathcal{V})$ such that $\| \mathbf{u}(\mathcal{V}) \| = 1$ and $\mathbf{u}(\mathcal{V}) \perp \mathbf{h}_k$ for all $k \in [K] \backslash \mathcal{V}$, we can build the transmission vector as in~\eqref{eq:transmission_vector_power_optimization}. 
Alternatively, instead of ZF we can use the optimized beamformer as in~\eqref{eq:transmission_vector_beamformer}. In this case, $X(\mathcal{V})$ is considered as Gaussian interference instead of being nulled at users $k \in [K] \backslash \mathcal{V}$, and the resulting symmetric rate will be higher~\cite{tolli2017multi,tolli2018linear}.

According to the above explanation, all users in $\mathcal{V}$ can decode part of their requested data after receiving $X(\mathcal{V})$; and as each missing data part is addressed in one set $\mathcal{V}$, building $\mathbf{x}$ as~\eqref{eq:transmission_vector_power_optimization} enables all users to decode every missing data part and completely decode their requested files (in Example~\ref{exmp:eff_ind} this procedure is reviewed for a simple network). It should be noted that each user faces a MAC channel after transmission of $\mathbf{x}$, and transmission rate and power should be adjusted that simultaneous decoding of all missing data parts is possible.



\subsection{Efficiency Index}
Assume the transmission vector $\mathbf{x}$ is built as in~\eqref{eq:transmission_vector_power_optimization}\footnote{For notational convenience we assume ZF beamforming, but the following discussion also holds for the optimized beamformer~\eqref{eq:transmission_vector_beamformer}.}. After $\mathbf{x}$ is transmitted, user $k$ receives
\begin{equation}
    y(k) = \mathbf{h}_k^T \mathbf{x} + z_k
\end{equation}
where 
$z_k$ is the additive Gaussian noise at user $k$\footnote{For $\mathbf{x}$ built as in~\eqref{eq:transmission_vector_beamformer}, $z_k$ would also contain Gaussian interference terms from all $\mathbf{w}(\mathcal{V})X(\mathcal{V}) \mid k \not\in \mathcal{V}$~\cite{tolli2017multi,tolli2018linear}.}. For each constituting term of $\mathbf{x}$, one of the following is possible:
\begin{enumerate}
    \item $k \not\in \mathcal{V}$:  $\mathbf{h}_k \perp \mathbf{u}(\mathcal{V})$ and hence the term $X(\mathcal{V})$ is nulled (or suppressed by $\mathbf{w}(\mathcal{V})$) at user $k$;
    \item $k \in \mathcal{V}$ and $\exists \; p \in \Phi(\mathcal{V}) \mid v_{p,k} = 0$: the term $X(\mathcal{V})$ is received, from which  the user $k$ can then decode the desired part $W_p(k)$;
    \item $k \in \mathcal{V}$ and $\nexists \; p \in \Phi(\mathcal{V}) \mid v_{p,k} = 0$: the term $X(\mathcal{V})$ is received at user $k$, but does not contain any missing packet for user $k$ and so is entirely removed using its cache contents.
\end{enumerate}

The third possibility can be interpreted as a case of \textit{power loss}, as the power used to transmit the deleted terms is not used for increasing the signal quality at that particular user. This is the reasoning behind our definition of the efficiency index. Let $\theta(k)$ denote the number of terms completely removed by cache contents of user $k$. Then the efficiency index $\Gamma(k)$ is 
\begin{equation}
\label{eq:eff_index_main}
    \Gamma(k) = 1 - \frac{\theta(k)}{\phi(\mathbf{x})} \; ,
\end{equation}
where $\phi(\mathbf{x})$ is the total number of terms transmitted in transmission vector $\mathbf{x}$. 

\begin{exmp}
\label{exmp:eff_ind}
Consider the cache placement matrix $\mathbf{V}'$ in~\eqref{eq:V_mat_4_2_Full} and assume the demand set is $\{A,B,C,D\}$. After transmission of $\mathbf{x}$ is concluded, user 1 receives
\begin{equation}
    \label{eq:exmp_k4_y1}
    \begin{aligned}
        y(1) &= (B_3 \oplus D_2) \mathbf{h}_1^T \mathbf{u}_1 p_1  + (A_3 \oplus C_4) \mathbf{h}_1^T \mathbf{u}_2 p_2 \\
        +&(B_4 \oplus D_1) \mathbf{h}_1^T \mathbf{u}_3 p_3 + (A_2 \oplus C_1) \mathbf{h}_1^T \mathbf{u}_4 p_4 + z_1 \; ,
    \end{aligned}
\end{equation}
where $\mathbf{u}_k = \mathbf{u}(\mathcal{V}_k)$, $p_k = p(\mathcal{V}_k)$ and $\mathcal{V}_k = [4] \backslash \{k \}$ (as $t=2$, each $\mathcal{V}$ includes three users). Considering RHS of \eqref{eq:exmp_k4_y1}, $B_3 \oplus D_2$ in nulled (as $\mathbf{u}_1 \perp \mathbf{h}_1$), while $A_3 \oplus C_4$ and $A_2 \oplus C_1$ contain useful data ($A_2$ and $A_3$; $C_4$ and $C_1$ are removed by the cache). However, $B_4 \oplus D_1$ is entirely removed using the cache content of user 1. Therefore, $\phi(\mathbf{x}) = 4$, $\theta(1) = 1$ and the efficiency index of user 1 becomes $\Gamma(1) = 0.75$. It can be easily checked that the efficiency index of other users is the same and independent of the request vector.
\end{exmp}

\subsection{Building the Placement Matrix}
\label{subsec:placement_matrix_build}
In order to build valid placement matrices, we introduce \textit{placement blocks}. Denoted by $\mathbf{\Lambda}_t^K(p)$, a placement block of size $p$ for $K$ users with global memory ratio $t$ is a $p \times K$ binary matrix, for which:
\begin{enumerate}
    \item Placement validity conditions \eqref{eq:per_user_cache_size} and \eqref{eq:per_part_cache_size} are met;
    \item Each row can be transformed into any other row using circular shift operations.
\end{enumerate}
For example, $\mathbf{V}$ and $\mathbf{V}'$ in \eqref{eq:V_mat_4_2_Petros} and \eqref{eq:V_mat_4_2_Full} are $\mathbf{\Lambda}_2^4(2)$ and $\mathbf{\Lambda}_2^4(4)$, respectively. It should be noted that $\mathbf{\Lambda}_t^K(p)$ is not necessarily unique (examples will be provided in Section \ref{Section:Simulation_Results}). Moreover, it can be easily checked that the column-wise concatenation of two (or more) placement blocks results in a valid placement matrix. For example, concatenation of $\mathbf{V}$ and $\mathbf{V}'$ results in $\mathbf{V}''$ as mentioned in \eqref{eq:V_mat_4_2_Full}, which is a valid placement matrix. 

For general $K$ and $t$ values, one can easily build a placement block by considering a random $1 \times K$ vector with $t$ non-zero elements as the first row of the block, and building each new row by circular shifting previous row for one unit, until the first row is repeated. The resulting matrix is a placement block. These placement blocks can then be concatenated freely to build valid placement matrices with different $P$ values.

\subsection{Extension to $K > t+L$}
So far we have only considered networks with $K = t+L$. If $K < t+L$, as the number of users falls below the maximum multicasting gain, providing an appropriate coded caching scheme becomes very straightforward. Specifically, any solution for a network with $L' = K-t$ antennas works for the same network with $L > L'$ antennas. Moreover, as mentioned in \cite{tolli2017multi, tolli2018linear} the extra antennas provide increased flexibility and gain for the multicast beamformer design.

In case $K > t+L$ the solution is not trivial. Defining $Q$ as~\eqref{eq:Q_def} 
we provide a scheme with total subpacketization $P \times Q$, where $P$ can be any number for which a valid placement matrix of dimensions $P \times K$ and global cache ratio $t$ exists.

During the placement phase each file $W$ is first split into $P$ packets denoted by $W_p$, and then each packet is further split into $Q$ subpackets $W_p^q$. Assume $\mathbf{V}$ is the given valid placement matrix and $W$ is a file in the library. For each $k \in [K]$ and $p \in [P]$, if $v_{p,k} = 1$ we store $W_p^q$ for all $q \in [Q]$ in the cache memory of user $k$.

For delivery, we select all subsets $\mathcal{S} \subset [K]$ with $|\mathcal{S}| = t+L$. For each $\mathcal{S}$, we generate matrix $\mathbf{V}^{\mathcal{S}}$ of size $P \times K$ and elements $v^{\mathcal{S}}_{p,k}$, using the procedure of Alg.~\ref{alg:newplacement}. $\mathbf{V}^{\mathcal{S}}$ is then used as the input to Algorithm \ref{alg:delivery_general}, which executes exactly the same procedure described in Section \ref{subsec:delivery_scheme}, to build the transmission vector $\mathbf{x}(\mathcal{S})$. 

\begin{algorithm}
\caption{Generate $\mathbf{V}^{\mathcal{S}}$}
\label{alg:newplacement}
\begin{algorithmic}[1]
    \Procedure{NewPlacement}{$\mathbf{V}, \mathcal{S}$}
        \State $\mathcal{S}' \gets [K] \backslash \mathcal{S}$,
        $\Phi(\mathcal{S}') \gets 0$
        \ForAll{$p \in [P]$}
            \If{$v_{p,k} = 1, \; \forall k \in \mathcal{S}'$}
                \State $\Phi(\mathcal{S}') \gets \Phi(\mathcal{S}') \cup \{ p \}$
            \EndIf
        \EndFor
        \ForAll{$p \in \Phi(\mathcal{S}'), \; k \in [K]$}
            \State $v^{\mathcal{S}}_{p,k} \gets 1$
        \EndFor
        \ForAll{$p \in [P], \; k \in \mathcal{S}'$}
            \State $v^{\mathcal{S}}_{p,k} \gets 1$
        \EndFor
        \ForAll{$p \in [P] \backslash \Phi(\mathcal{S}'), \; k \in \mathcal{S}$}
            \State $v^{\mathcal{S}}_{p,k} \gets v_{p,k}$
        \EndFor
    \EndProcedure
\end{algorithmic}
\end{algorithm}

\begin{algorithm}
\caption{Build Transmission Vector $\mathbf{x}(\mathcal{S})$}
\label{alg:delivery_general}
\begin{algorithmic}[1]
    \Procedure{BuildTX}{$\mathcal{S}, \mathbf{V^{\mathcal{S}}}, \{W(1),W(2),...,W(K) \}$}
            \State $\mathbf{x}(\mathcal{S}) \gets 0$
            \ForAll{$\mathcal{V} \subseteq \mathcal{S}$ with $|\mathcal{V}| = t+1$}
                \State $\Phi(\mathcal{V}) \gets \varnothing$
                \ForAll{$p \in [P]$}
                    \If{$v_{p,k} = 0, \; \forall k \in [K] \backslash \mathcal{V}$}
                        \State $\Phi(\mathcal{V}) \gets \Phi(\mathcal{V}) \cup \{ p \}$
                    \EndIf
                \EndFor
                \State $X(\mathcal{V}) \gets 0$
                \ForAll{$p \in \Phi(\mathcal{V}), k \in \mathcal{V}$}
                    \If{$v^{\mathcal{S}}_{p,k} = 0$}
                        \State $X(\mathcal{V}) \gets X(\mathcal{V}) \oplus W_p^{q(W(k),k)}(k)$
                        \State $q(W(k),k) \gets q(W(k),k) + 1$
                    \EndIf
                \EndFor
                \State $\mathbf{x}(\mathcal{S}) \gets \mathbf{x}(\mathcal{S}) + X(\mathcal{V}) \times \mathbf{u}(\mathcal{V}) \times p(\mathcal{V})$
            \EndFor
    \EndProcedure
\end{algorithmic}
\end{algorithm}

Note that $\mathbf{V}^{\mathcal{S}}$ is generated such that no term in $\mathbf{x}(\mathcal{S})$ needs to be nulled (or suppressed) at more than $L-1$ users. In Algorithm \ref{alg:delivery_general}, the initial values of superscripts $q(n,k)$ are all set to one.

\begin{exmp}
Assume $K=5$, $t=L=2$ and the file library is $\{ A,B,C,D,E \}$. The scheme of \cite{shariatpanahi2017multi} requires subpacketization $\binom{5}{2} \binom{2}{1} = 20$, but using the cache placement matrix
\begin{equation}
    \mathbf{V} = 
    \begin{bmatrix}
        1 & 1 & 0 & 0 & 0 \\
        0 & 1 & 1 & 0 & 0 \\
        0 & 0 & 1 & 1 & 0 \\
        0 & 0 & 0 & 1 & 1 \\
        1 & 0 & 0 & 0 & 1
    \end{bmatrix} \; ,
\end{equation}
the required subpacketization is reduced to $5 \times \binom{2}{1} = 10$. For the example demand set $\{ A,B,C,D,E \}$ we need $\binom{5}{4} = 5$ consecutive transmissions, which are built as
\begin{equation*}
\begin{aligned}
    \mathbf{x}(\mathcal{S}_1) &= (D_2^1 \oplus B_3^1) \mathbf{u}_5 + (E_3^1 \oplus C_4^1) \mathbf{u}_2 + E_2^1 \mathbf{u}_4 + B_4^1 \mathbf{u}_3 \\
    \mathbf{x}(\mathcal{S}_2) &= (E_3^2 \oplus C_4^2) \mathbf{u}_1 + (A_4^1 \oplus D_5^1) \mathbf{u}_3 + A_3^1 \mathbf{u}_5 + C_5^1 \mathbf{u}_4 \\
    \mathbf{x}(\mathcal{S}_3) &= (A_4^2 \oplus D_5^2) \mathbf{u}_2 + (E_1^1 \oplus B_5^1) \mathbf{u}_4 + B_4^2 \mathbf{u}_1 + D_1^1 \mathbf{u}_5 \\
    \mathbf{x}(\mathcal{S}_4) &= (E_1^2 \oplus B_5^2) \mathbf{u}_3 + (C_1^1 \oplus A_2^1) \mathbf{u}_5 + C_5^2 \mathbf{u}_2 + E_2^2 \mathbf{u}_1 \\
    \mathbf{x}(\mathcal{S}_5) &= (C_1^2 \oplus A_2^2) \mathbf{u}_4 + (D_2^2 \oplus B_3^2) \mathbf{u}_1 + D_1^2 \mathbf{u}_3 + A_3^2 \mathbf{u}_2 \\
\end{aligned}
\end{equation*}
where $\mathcal{S}_k = [K] \backslash \{ k \}$, $\mathbf{u}_k$ is defined such that $\| \mathbf{u}_k \| = 1$ and $\mathbf{u}_k \perp \mathbf{h}_k$, and we have ignored power coefficients for notational simplicity. It can be easily checked that after all the transmissions are concluded, all users can decode their requested files.
\end{exmp}

\section{Simulation Results}
\label{Section:Simulation_Results}
We provide simulation results for two network setups, with parameters shown in Table \ref{tab:Sim_Param}. In order to build valid placement matrices we use placement block concatenation, as described in Section \ref{subsec:placement_matrix_build}. Due to lack of space we explicitly explain the procedure for Network 1 only. For this network, there exist two placement blocks of size $6 \times 6$, which are
\begin{equation*}
    \mathbf{V}_1 = 
    \begin{bmatrix}
        1 & 1 & 0 & 0 & 0 & 0 \\
        0 & 1 & 1 & 0 & 0 & 0 \\
        0 & 0 & 1 & 1 & 0 & 0 \\
        0 & 0 & 0 & 1 & 1 & 0 \\
        0 & 0 & 0 & 0 & 1 & 1 \\
        1 & 0 & 0 & 0 & 0 & 1 \\
    \end{bmatrix},
    \mathbf{V}_2 =
    \begin{bmatrix}
        1 & 0 & 1 & 0 & 0 & 0 \\
        0 & 1 & 0 & 1 & 0 & 0 \\
        0 & 0 & 1 & 0 & 1 & 0 \\
        0 & 0 & 0 & 1 & 0 & 1 \\
        1 & 0 & 0 & 0 & 1 & 0 \\
        0 & 1 & 0 & 0 & 0 & 1 \\
    \end{bmatrix},
\end{equation*}
and another placement block of size 3: 
\begin{equation*}
    \mathbf{V}_3 = 
    \begin{bmatrix}
        1 & 0 & 0 & 1 & 0 & 0 \\
        0 & 1 & 0 & 0 & 1 & 0 \\
        0 & 0 & 1 & 0 & 0 & 1 \\
    \end{bmatrix} \; .
\end{equation*}
$\mathbf{V}_1$ and $\mathbf{V}_2$ are confirming examples that placement blocks of a given size are not necessarily unique. Using $\mathbf{V}_1$, $\mathbf{V}_2$, $\mathbf{V}_3$ and their concatenations, one can create valid placement matrices with $P$ values mentioned in Table~\ref{tab:Sim_Param}\footnote{Although for Network 1 all possible $P$ values are mentioned, the $P$ values considered for Network 2 are a subset of all possible values. For example, there exists a placement with $P=8$, which is not considered.}. For example, $P=9$ can be created by concatenating either $\mathbf{V}_1$ or $\mathbf{V}_2$ with $\mathbf{V}_3$, and $P=12$ can be created by concatenating $\mathbf{V}_1$ and $\mathbf{V}_2$. The case $P=15$ is equivalent to the scheme of \cite{shariatpanahi2017multi} and is created by concatenating all the three blocks. It is worth mentioning that while the scheme presented in \cite{lampiris2018adding} is not readily applicable for Network 1 (as $\frac{K}{L}$ is not an integer), our scheme easily provides a cache placement with subpacketization $P=3$.

\begin{table}[H]
    \centering
    \begin{tabular}{|c||c|c|}
        \hline
        & Network 1 & Network 2 \\
        \hline \hline
        $K$ & 6 & 6  \\
        \hline
        $t$ & 2 & 3 \\
        \hline
        $L$ & 4 & 3 \\
        \hline
        $P$ values & $\{3,6,9,12,15 \}$ & $\{2,6,12,18,20 \}$ \\
        \hline
    \end{tabular}
    \caption{\label{tab:Sim_Param} Simulation Parameters}
\end{table}

In Figures \ref{fig:k6_t2_plot} and \ref{fig:k6_t3_plot} we have plotted simulation results for Networks 1 and 2, respectively. For simulations we have used optimized beamformers as in~\eqref{eq:transmission_vector_power_optimization}~\cite{tolli2017multi}. Moreover, efficiency indices corresponding to different $P$ values for both networks are presented in Table \ref{tab:eff_index}. It can be verified that higher subpacketization usually\footnote{There exist special cases for which higher subpacketization reduces the efficiency index. This is not covered here due to lack of space.} results in better efficiency index; which in turn means higher symmetric rate. Also, as the slopes are the same for all curves in Figures \ref{fig:k6_t2_plot} and \ref{fig:k6_t3_plot}, it is clear that selection of $P$ is not affecting the multicasting gain.

\begin{figure}
    \centering
    \resizebox{0.9\columnwidth}{!}{%

    \begin{tikzpicture}

    \begin{axis}
    [
    axis lines = left,
    xlabel = \smaller {SNR [dB]},
    ylabel = \smaller {Symmetric Rate [nats/s]},
    ylabel near ticks,
    legend pos = south east,
    ticklabel style={font=\smaller},
    grid=both,
    major grid style={line width=.2pt,draw=gray!30},
    ]
    
    \addplot
    [
    mark = +,
    black
    ]
    table
    [
    y=P15,
    x=SNR
    ]
    {Data/K6t2_Plot.txt};
    \addlegendentry{\smaller $P = 15$}
    
    \addplot
    [
    dashed,
    mark = +,
    gray
    ]
    table
    [
    y=P12,
    x=SNR
    ]
    {Data/K6t2_Plot.txt};
    \addlegendentry{\smaller $P = 12$}
    
    \addplot
    [
    mark = x,
    black
    ]
    table
    [
    y=P9,
    x=SNR
    ]
    {Data/K6t2_Plot.txt};
    \addlegendentry{\smaller $P = 9$}
    
    \addplot
    [
    dashed,
    mark = x,
    black
    ]
    table
    [
    y=P6,
    x=SNR
    ]
    {Data/K6t2_Plot.txt};
    \addlegendentry{\smaller $P = 6$}
    
    \addplot
    [
    mark = square,
    black
    ]
    table
    [
    y=P3,
    x=SNR
    ]
    {Data/K6t2_Plot.txt};
    \addlegendentry{\smaller $P = 3$}
    
    \end{axis}

    \end{tikzpicture}
    }

    \caption{Rate vs SNR, Optimized Beamformer - $K=6$, $t=2$, $L=4$}
    \label{fig:k6_t2_plot}
\end{figure}
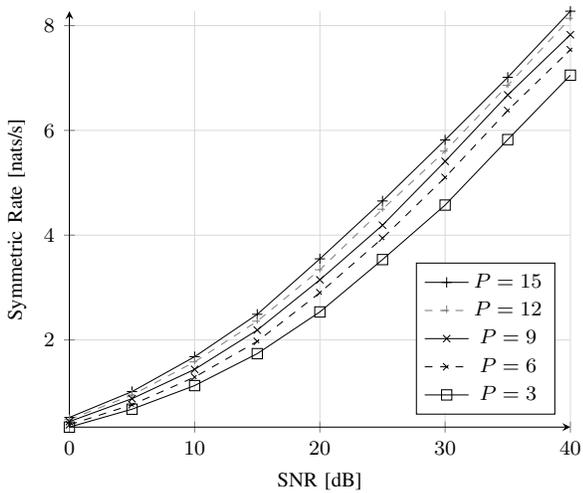

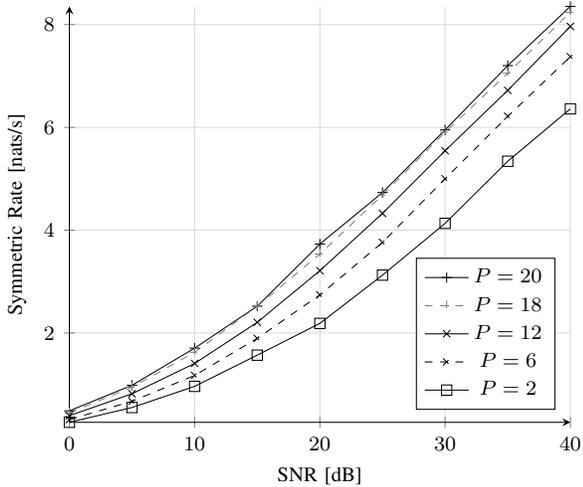
\begin{figure}
    \centering
    \resizebox{0.9\columnwidth}{!}{%
        \begin{tikzpicture}

    \begin{axis}
    [
    axis lines = left,
    xlabel = \smaller {SNR [dB]},
    ylabel = \smaller {Symmetric Rate [nats/s]},
    ylabel near ticks,
    legend pos = south east,
    ticklabel style={font=\smaller},
    grid=both,
    major grid style={line width=.2pt,draw=gray!30},
    ]
    
    \addplot
    [
    mark = +,
    black
    ]
    table
    [
    y=P20,
    x=SNR
    ]
    {Data/K6t3_Plot.txt};
    \addlegendentry{\smaller $P = 20$}
    
    \addplot
    [
    dashed,
    mark = +,
    gray
    ]
    table
    [
    y=P18,
    x=SNR
    ]
    {Data/K6t3_Plot.txt};
    \addlegendentry{\smaller $P = 18$}
    
    \addplot
    [
    mark = x,
    black
    ]
    table
    [
    y=P12,
    x=SNR
    ]
    {Data/K6t3_Plot.txt};
    \addlegendentry{\smaller $P = 12$}
    
    \addplot
    [
    dashed,
    mark = x,
    black
    ]
    table
    [
    y=P6,
    x=SNR
    ]
    {Data/K6t3_Plot.txt};
    \addlegendentry{\smaller $P = 6$}
    
    \addplot
    [
    mark = square,
    black
    ]
    table
    [
    y=P2,
    x=SNR
    ]
    {Data/K6t3_Plot.txt};
    \addlegendentry{\smaller $P = 2$}
    
    \end{axis}

    \end{tikzpicture}
    }
    \caption{Rate vs SNR, Optimized Beamformer - $K=6$, $t=3$, $L=3$}
    \label{fig:k6_t3_plot}
\end{figure}

\begin{table}[H]
    \centering
    \begin{tabular}{|c||c|c|c|c|c|}
        \hline
        \multirow{2}{*}{Network 1} & $P=3$ & $P=6$ & $P=9$ & $P=12$ & $P = 15$ \\
        \cline{2-6}
        & 0.667 & 0.722 & 0.833 & 0.90 & 1.00 \\
        \hline
        \hline
        \multirow{2}{*}{Network 2} & $P=2$ & $P=6$ & $P=12$ & $P=18$ & $P = 20$ \\
        \cline{2-6}
        & 0.50 & 0.583 & 0.733 & 0.933 & 1.00 \\
        \hline
    \end{tabular}
    \caption{Efficiency Index for Various $P$ Values}
    \label{tab:eff_index}
\end{table}

Finally, in Figure \ref{fig:k6_t2_bar} we have plotted the rate advantage of various subpacketization levels over $P=3$, versus SNR for Network 1 (similar results hold for Network 2). It is clear that higher subpacketization results in better rate advantage, specially at low-SNR.

\begin{figure}
    \centering
    \resizebox{0.9\columnwidth}{!}{%
   
    \begin{tikzpicture}
    \begin{axis}
    [
    axis lines = left,
    xlabel = \smaller {SNR [dB]},
    ylabel = \smaller {Rate Advantage [\%]},
    ylabel near ticks,
    legend pos = north east,
    ticklabel style={font=\smaller},
    ybar interval=0.7,
    ]
    
    \addplot
    [
    black,
    fill = gray!10,
    ]
    table
    [x=SNR,y=P15]
    {Data/K6t2_Bar.txt};
    \addlegendentry{\smaller $P=15$}

    \addplot 
    [
    black,
    fill = gray!30,
    ]
    table
    [x=SNR,y=P12]
    {Data/K6t2_Bar.txt};
    \addlegendentry{\smaller $P=12$}
    
    \addplot 
    [
    black,
    fill = gray!50,
    ]
    table
    [x=SNR,y=P9]
    {Data/K6t2_Bar.txt};
    \addlegendentry{\smaller $P=9$}
    
    \addplot 
    [
    black,
    fill = gray!70,
    ]
    table
    [x=SNR,y=P6]
    {Data/K6t2_Bar.txt};
    \addlegendentry{\smaller $P=6$}
    
    \end{axis}
    \end{tikzpicture}
    }
    \caption{Rate Advantage over $P=3$ - $K=6$, $t=2$, $L=4$}
    \label{fig:k6_t2_bar}
\end{figure}
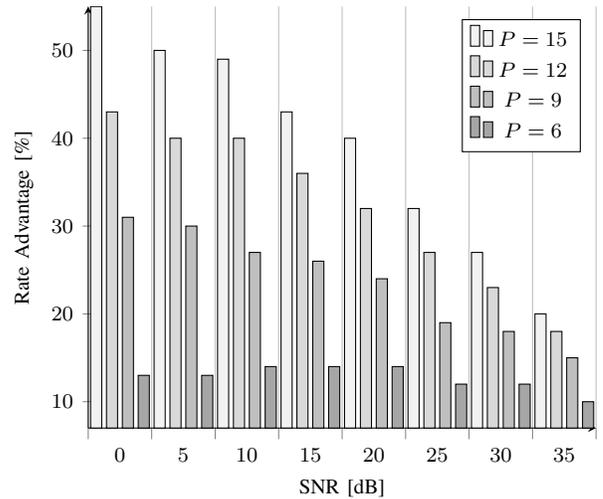

\section{Conclusion and Future Work}
\label{Section:conclusion}
We provided a new scheme for multi-antenna coded caching, which enables the subpacketization level to be selected freely among a set of predefined values depending on basic network parameters such as user count, antenna count and global cache ratio. We also proposed a simple efficiency index as a performance indicator (in terms of symmetric rate) at various subpacketization levels. Numerical examples demonstrate that larger subpacketization generally results in larger efficiency index and higher symmetric rate, while low subpacketization incurs significant loss in the achievable rate. Our scheme enables more efficient caching decisions, tailored to the available computational and power resources.

Possible future research directions include finding upper bounds for the achievable symmetric rate for various subpacketization levels, further simplifying the scheme for large networks, and thorough comparison of various beamformer design techniques (ZF versus optimized design versus approximate solutions).

\bibliographystyle{IEEEtran}
\bibliography{main}

\end{document}